\documentclass{llncs}
\title{Shepherding Hordes of Markov Chains\thanks{This work has been supported by
the DFG RTG 2236 ``UnRAVeL'' and the Czech Science Foundation grant No. Robust 17-12465S}}
\author{Milan \v{C}e\v{s}ka\inst{1} \and Nils Jansen\inst{2} \and Sebastian Junges\inst{3} \and Joost-Pieter Katoen\inst{3}
}
\institute{
Brno University of Technology, Brno, Czech Republic \and 
Radboud University, Nijmegen, The Netherlands \and
RWTH Aachen University, Aachen, Germany
}
\usepackage{xcolor}
\usepackage{amsmath,amssymb,amsfonts}
\usepackage{mathtools}
\usepackage{xspace}
\usepackage{subfigure}
\usepackage{paralist}
\usepackage{microtype}
\usepackage{colonequals}
\usepackage{tikz}
\usepackage{cite}
\usepackage{multirow}
\usetikzlibrary{arrows.meta,decorations.pathmorphing,positioning,fit,trees,shapes,shadows,automata,calc,decorations.markings} 
\tikzset{>=latex}

\newcommand{\init}{\ensuremath{s_0}}
\newtheorem{mydef}{Definition}

\newtheorem{mytheorem}{Theorem}

\newtheorem{myproblem}{Problem}
\newtheorem{mycorollary}{Corollary}
\newtheorem{mymethod}{Approach}

\usepackage{algorithmicx,algorithm}
\usepackage[noend]{algpseudocode}

\newcommand{\pathset}{\mathsf{Paths}}

\newcommand{\pathsfin}{\pathset_{\mathit{fin}}}

\newcommand{\act}{\ensuremath{a}}
\newcommand{\Act}{\ensuremath{\mathit{Act}}}
\newcommand{\last}[1]{\mathrm{last}(#1)}
\DeclareMathOperator{\supp}{supp}
\DeclareMathOperator{\successors}{succ}
\newcommand{\Distr}{\mathit{Distr}}

\newcommand{\distDom}{X}
\newcommand{\distFunc}{\mu}
\newcommand{\distDomElem}{x}
\newcommand{\true}{\texttt{true}\xspace}

\newcommand{\probOneG}{\mathit{p1G}}
\newcommand{\probPosG}{\mathit{pPG}}
\newcommand{\quotientstates}{\ensuremath{S^\mathfrak{D}_\sim}}

\newcommand{\dtmc}{\ensuremath{D}}
\newcommand{\mdp}{\ensuremath{M}}

\begin{document}
\maketitle

\begin{abstract}

This paper considers large families of Markov chains (MCs) that are
defined over a set of parameters with finite discrete domains. Such
families occur in software product lines, planning under partial
observability, and sketching of probabilistic programs. Simple questions,
like `does at least one family member satisfy a property?', are NP-hard.
We tackle two problems: distinguish family members that satisfy a given quantitative property from those that do not, and determine
a family member that satisfies the property optimally, i.e., with the highest probability or reward.
We show that combining two well-known techniques, MDP model checking and abstraction refinement, mitigates the
computational complexity. Experiments on a broad set of benchmarks show
that in many situations, our approach is able to handle families of millions of MCs,
providing superior scalability compared to existing~solutions.

\end{abstract}

\section{Introduction}

%
Randomisation is key to research fields such as dependability (uncertain system components), distributed computing (symmetry breaking), planning (unpredictable environments), and probabilistic programming.
Families of alternative designs differing in the structure and system parameters are ubiquitous.
Software dependability has to cope with configuration options, in distributed computing the available memory per process is highly relevant, in planning the observability of the environment is pivotal, and program synthesis is all about selecting correct program variants.
The automated analysis of such families has to face a formidable challenge --- in addition to the state-space explosion affecting each family member, the family size typically grows exponentially in the number of features, options, or observations.
This affects the analysis of (quantitative) software product lines~\cite{GS13,DBLP:conf/splc/VarshosazK13,RodriguesANLCSS15,DBLP:conf/fm/VandinBLL18,DBLP:journals/fac/ChrszonDKB18}, strategy synthesis in planning under partial
observability~\cite{Koc2015,DBLP:conf/aaai/ChatterjeeCD16,DBLP:conf/aaai/ChadesCMNSB12,DBLP:journals/rts/Norman0Z17,DBLP:journals/tcs/GiroDF14}, and probabilistic program synthesis~\cite{DBLP:conf/pldi/NoriORV15,DBLP:conf/cav/ChasinsP17,DBLP:journals/ase/GerasimouCT18,CALINESCU2018140}. 

%
This paper considers families of Markov chains (MCs) to describe configurable probabilistic systems. 
We consider finite MC families with finite-state family members.
Family members may have different transition probabilities and distinct topologies --- thus different reachable state spaces.
The latter aspect goes beyond the class of parametric MCs as considered in parameter synthesis~\cite{dehnert2015prophesy,hahn2011probabilistic,Ceska2017,DBLP:conf/atva/CubuktepeJJKT18} and model repair~\cite{DBLP:conf/tase/ChenHHKQ013,bartocci2011model,pathak-et-al-nfm-2015}. 

For an MC family $\mathfrak{D}$ and quantitative specification $\varphi$, with $\varphi$ a reachability probability or expected reward objective, we consider the following synthesis problems: (a) does some member in $\mathfrak{D}$ satisfy a threshold on~$\varphi$? (aka: \emph{feasibility} synthesis), (b) which members of $\mathfrak{D}$ satisfy this threshold on $\varphi$ and which ones do not? (aka: \emph{threshold synthesis}), and (c) which family member(s) satisfy $\varphi$ optimally, e.g., with highest probability? (aka: \emph{optimal synthesis}).

%
The simplest synthesis problem, feasibility, is NP-complete and can naively be solved by analysing all individual family members --- the so-called \emph{one-by-one} approach. 
This approach has been used in \cite{DBLP:journals/fac/ChrszonDKB18} (and for qualitative systems in e.g.\ \cite{DBLP:journals/sttt/ClassenCHLS12}), but is infeasible for large systems.
An alternative is to model the family $\mathfrak{D}$ by a single Markov decision process (MDP) --- the so-called \emph{all-in-one} MDP~\cite{DBLP:journals/fac/ChrszonDKB18}. 
The initial MDP state non-deterministically chooses a family member of $\mathfrak{D}$, and then evolves in the MC of that member.
This approach has been implemented in tools such as ProFeat~\cite{DBLP:journals/fac/ChrszonDKB18}, and for purely qualitative systems in~\cite{DBLP:journals/scp/ClassenCHLS14}.
The MDP representation avoids the individual analysis of all family members, but its size is proportional to the family size. 
This approach therefore does not scale to large families.
A symbolic BDD-based approach is only a partial solution as family members may induce different reachable state-sets.

This paper introduces an \emph{abstraction-refinement} scheme over the MDP representation\footnote{Classical CEGAR for model checking of software product lines has been proposed in~\cite{DBLP:conf/sigsoft/CordyHLSDL14}. This uses feature transition systems, is purely qualitative, and exploits existential state abstraction.}.
The abstraction \emph{forgets} in which family member the MDP operates.
The resulting \emph{quotient} MDP has a single representative for every reachable state in a family member. 
It typically provides a very compact representation of the family $\mathfrak{D}$ and its analysis using off-the-shelf MDP model-checking algorithms yields a speed-up compared to the all-in-one approach. 
Verifying the quotient MDP yields under- and over-approximations of the $\min$ and $\max$ probability (or reward), respectively. 
These bounds are safe as all \emph{consistent} schedulers, i.e., those that pick actions according to a single family member, are contained in all schedulers considered on the quotient MDP.
%
(CEGAR-based MDP model checking for partial information schedulers, a slightly different notion than restricting schedulers to consistent ones, has been considered in~\cite{DBLP:conf/atva/GiroR12}. In contrast to our setting, \cite{DBLP:conf/atva/GiroR12} considers history-dependent schedulers and in this general setting no guarantee can be given that bounds on suprema converge~\cite{DBLP:journals/tcs/GiroDF14}).

Model-checking results of the quotient MDP do provide useful insights. 
This is evident if the resulting scheduler is consistent.
If the verification reveals that the $\min$ probability exceeds $r$ for a specification $\varphi$ with a $\leq r$ threshold, then --- even for inconsistent schedulers --- it holds that all family members violate~$\varphi$. 
If the model checking is inconclusive, i.e., the abstraction is too coarse, we iteratively refine the quotient MDP by splitting the family into sub-families. 
We do so in an efficient manner that avoids rebuilding the sub-families.
Refinement employs a light-weight analysis of the model-checking results. 

We implemented our abstraction-refinement approach using the Storm model checker~\cite{DBLP:conf/cav/DehnertJK017}.
Experiments with case studies from software product lines, planning, and distributed computing yield possible speed-ups of up to 3 orders of magnitude over the one-by-one and all-in-one approaches (both symbolic and explicit). 
Some benchmarks include families of millions of MCs where family members are thousands of states.
The experiments reveal that --- as opposed to parameter synthesis~\cite{dehnert2015prophesy,hahn2011probabilistic,Ceska2017} --- the threshold has a major influence on the synthesis times.

To summarise, this work presents:
a)~MDP-based abstraction-refinement for various synthesis problems over large families of MCs,
b)~a refinement strategy that mitigates the overhead of analysing sub-families, 
and c)~experiments showing substantial speed-ups for many benchmarks.
Extra material can be found in~\cite{additionalmaterial}.

\section{Preliminaries}
We present the basic foundations for this paper, for details, we refer to~\cite{BK08,DBLP:reference/mc/BaierAFK18}.

\paragraph{Probabilistic models.}
A \emph{probability distribution} over a finite or countably infinite set $\distDom$
is a function $\distFunc\colon\distDom\rightarrow [0,1]$ with $\sum_{\distDomElem\in\distDom}\distFunc(\distDomElem)=\distFunc(\distDom)=1$.
The set of all distributions on $\distDom$ is denoted $\Distr(\distDom)$. 
The support of a distribution $\distFunc$ is $\supp(\distFunc) = \{x\in\distDom\,|\,\distFunc(x)>0\}$.
A distribution is \emph{Dirac} if $|\!\supp(\distFunc)| = 1$.

\begin{mydef}[MC]\label{def:dtmc}
A \emph{discrete-time Markov chain} (MC) $D$ is a triple $(S,\init, \mathbf{P})$, where $S$ is a finite set of states, $\init \in S$ is an initial state, and
$\mathbf{P}\colon S \rightarrow \Distr(S)$ is a transition probability matrix.
\end{mydef}
MCs have unique distributions over successor states at each state. 
Adding nondeterministic choices over distributions leads to Markov decision processes.
\begin{mydef}[MDP]\label{def:mdp}	
A \emph{Markov decision process} (MDP) is a tuple $M=(S,\init,\Act,\mathcal{P})$ where $S, \init$ as in Def.~\ref{def:dtmc}, $\Act$ is a finite set of actions, and 
$\mathcal{P}\colon S\times \Act \nrightarrow \Distr(S)$ is a partial transition probability function.
\end{mydef}

The \emph{available actions} in $s\in S$ are $\Act(s)=\{\act\in\Act\mid \mathcal{P}(s,\act) \neq \bot \}$. An MDP with $|\Act(s)|=1$ for all $s\in S$ is an MC. 
For MCs (and MDPs), a state-reward function is $\textit{rew}\colon S\rightarrow \mathbb{R}_{\geq 0}$.
The reward $\textit{rew}(s)$ is earned upon leaving~$s$.

A \emph{path} of an MDP $M$ is an (in)finite sequence $\pi = s_0\xrightarrow{\act_0}s_1\xrightarrow{\act_1}\cdots$,
where $s_i\in S$, $\act_i\in\Act(s_i)$, and $\mathcal{P}(s_i,\act_i)(s_{i+1})\neq 0$ for all $i\in\mathbb{N}$.
For finite $\pi$, $\last{\pi}$ denotes
the last state of $\pi$. The set of (in)finite paths of $M$ is $\pathsfin^{M}$ ($\pathset^{M}$).
The notions of paths carry over to MCs (actions are omitted).
Schedulers resolve all choices of actions in an MDP and yield MCs.
\begin{mydef}[Scheduler]	\label{def:scheduler}
A \emph{scheduler} for an MDP $\mdp=(S,\init,\mathit{Act},\mathcal{P})$ is a function $\sigma\colon\pathsfin^{M}\to\Act$
  such that $\sigma(\pi)\in \Act(\last{\pi})$ for all $\pi\in \pathsfin^{\mdp}$.
   Scheduler $\sigma$ is \emph{memoryless} if $\last{\pi}=\last{\pi'} \implies \sigma(\pi)=\sigma(\pi')$ for all $\pi,\pi'\in\pathsfin^{\mdp}$.
  The set of all schedulers of $\mdp$ is $\Sigma^{\mdp}$.
\end{mydef}
\begin{mydef}[Induced Markov Chain]
  \label{def:induced_dtmc}
The MC \emph{induced by MDP $\mdp$ and $\sigma \in \Sigma^\mdp$} is given by $\mdp_\sigma = (\pathsfin^{\mdp},\init,\mathbf{P}^{\sigma})$ where:
  \[
    \mathbf{P}^{\sigma}(\pi,\pi') = \begin{cases}
        \mathcal{P}(\last{\pi},\sigma(\pi))(s') & \text{if }\pi' = \pi\xrightarrow{\sigma(\pi)} s' \\
        0 & \text{otherwise.}
      \end{cases}
  \]
\end{mydef}
\paragraph{Specifications.} For a MC $\dtmc$, we consider unbounded reachability specifications of the form $\varphi=\mathbb{P}_{\sim \lambda}(\lozenge G)$ with $G\subseteq S$ a set of goal states, $\lambda\in [0,1]\subseteq{\mathbb{R}}$, and ${{}\sim{}} \in \{<,\leq,\geq,>\}$.
The probability to satisfy the path formula $\phi=\lozenge G$ in $\dtmc$ is denoted by $\texttt{Prob}(\dtmc, \phi)$.
If $\varphi$ holds for $\dtmc$, that is, $\texttt{Prob}(\dtmc, \phi)\sim \lambda$, we write $ \dtmc\models\varphi$.
Analogously, we define expected reward specifications of the form $\varphi=\mathbb{E}_{\sim \kappa}(\lozenge G)$ with $\kappa\in {\mathbb{R}_{\geq 0}}$. We refer to $\lambda$/$\kappa$ as \emph{thresholds}.
While we only introduce reachability specifications, our approaches may be extended to richer logics like arbitrary PCTL~\cite{hansson_pctl}, PCTL*~\cite{aziz1995usually}, or $\omega$-regular properties.

For an MDP $\mdp$, a specification $\varphi$ holds ($M\models\varphi$) if and only if it holds for the induced MCs of all schedulers.
The maximum probability $\texttt{Prob}^{\max}(\mdp,\phi)$ to satisfy a path formula $\phi$ for an MDP $M$ is given by a maximising scheduler $\sigma^{\max}\in\Sigma^M$, that is, there is no scheduler $\sigma'\in\Sigma^M$ such that $\texttt{Prob}(\mdp_{\sigma^{\max}},\phi)<\texttt{Prob}(M_{\sigma'},\phi)$.
Analogously, we define the minimising probability $\texttt{Prob}^{\min}(\mdp,\phi)$, and the maximising (minimising) expected reward $\texttt{ExpRew}^{\max}(\mdp, \phi)$ ($\texttt{ExpRew}^{\min}(\mdp, \phi)$).

The probability (expected reward) to satisfy path formula $\phi$ from state $s\in S$ in MC $\dtmc$ is $\texttt{Prob}(\dtmc,\phi)(s)$ ($\texttt{ExpRew}(\dtmc,\phi)(s)$).
The notation is analogous for maximising and minimising probability and expected reward measures in MDPs.
%
Note that the expected reward $\texttt{ExpRew}(\dtmc, \phi)$ to satisfy path formula $\phi$ is only defined if $\texttt{Prob}(\dtmc, \phi)=1$.
Accordingly, the expected reward for MDP $\mdp$ under scheduler $\sigma\in\Sigma^\mdp$ requires $\texttt{Prob}(\mdp_{\sigma}, \phi)=1$.	

\section{Families of MCs}
We present our approaches on the basis of an explicit representation of a \emph{family of MCs} using a parametric transition probability function.
While arbitrary probabilistic programs allow for more modelling freedom and complex parameter structures, the explicit representation alleviates the presentation and allows to reason about practically interesting synthesis problems. 
In our implementation, we use a more flexible high-level modelling language, cf. Sect~\ref{sec:implementation}.
%
\begin{mydef}[Family of MCs]
		A \emph{family of MCs} is defined as a tuple $\mathfrak{D}=(S,\init,K,\mathfrak{P})$ where $S$ is a finite set of states, $\init\in S$ is an initial state, $K$ is a finite set of discrete parameters such that the domain of each parameter $k\in K$ is $T_k\subseteq S$, and $\mathfrak{P}\colon S \rightarrow \Distr(K)$ is a family of transition probability matrices.
\end{mydef}
The transition probability function of MCs maps states to distributions over successor states. 
For families of MCs, this function maps states to distributions over parameters.
Instantiating each of these parameters with a value from its domain yields a ``concrete'' MC, called a \emph{realisation}.
%
%
\begin{mydef}[Realisation]
	A \emph{realisation} of a family $\mathfrak{D}=(S,\init,K,\mathfrak{P})$ is a function $r\colon K \rightarrow S$ where $\forall k\in K\colon r(k) \in T_k$. 
	A realisation $r$ yields a MC $D_r = (S,\init,\mathfrak{P}(r))$, where $\mathfrak{P}(r)$ is the transition probability matrix in which 
	each  
	$k\in K$ in $\mathfrak{P}$ 
	is replaced by $r(k)$.
	Let $\mathcal{R}^\mathfrak{D}$ denote the \emph{set of all realisations} for $\mathfrak{D}$.
\end{mydef}
As a family $\mathfrak{D}$ of MCs is defined over finite parameter domains, the number of family members (i.e. realisations from $\mathcal{R}^\mathfrak{D}$) of $\mathfrak{D}$ is finite, viz. $|\mathfrak{D}| \colonequals |\mathcal{R}^\mathfrak{D}|  = \prod_{k\in K}|T_k|$, but exponential in $|K|$.
Subsets of $\mathcal{R}^\mathfrak{D}$ induce so-called \emph{subfamilies} of $\mathfrak{D}$.
While all these MCs share the same state space, their \emph{reachable} states may differ, as demonstrated by the following example.

\begin{example}[Family of MCs]\label{ex:dtmc_family}
	Consider a family of MCs $\mathfrak{D} = (S,\init,K,\mathfrak{P})$ where $S=\{0,1,2,3\}$, $\init = 0$,  and $K=\{k_0,k_1,k_2\}$ with domains  $T_{k_0}=\{0\}, T_{k_1}=\{0,1\}$, and $T_{k_2}=\{2,3\}$. The parametric transition function $\mathfrak{P}$ is defined by:
		\begin{align*}
		&\mathfrak{P}(0) = 0.5\colon k_0 + 0.5\colon k_1 &\mathfrak{P}(1) = 0.5\colon k_1 + 0.5\colon k_2\\
		&\mathfrak{P}(2) = 1\colon k_2 				   &\mathfrak{P}(3) = 0.5\colon k_1 + 0.5\colon k_2
	\end{align*}
Fig.~\ref{fig:realisation} shows the four MCs that result from the realisations $\{r_1,r_2,r_3,r_4\}=\mathcal{R}^\mathfrak{D}$ of $\mathfrak{D}$.
States that are unreachable from the initial state are greyed out.
\end{example}
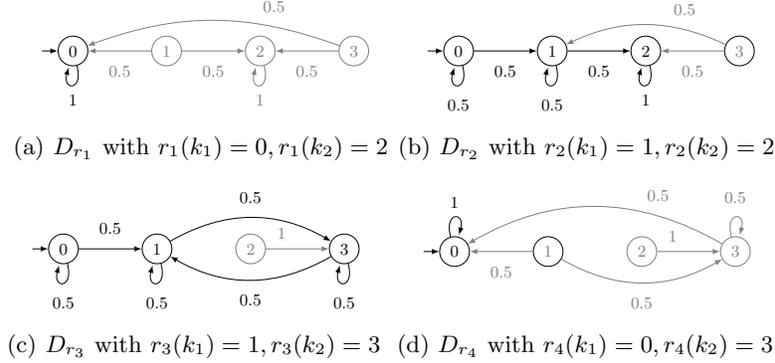
\begin{figure}[t]
	\centering
	\subfigure[$D_{r_1}$ with $r_1(k_1)=0, r_1(k_2)=2$]{%
	    \label{fig:realisation_0_2}
		\scalebox{0.7}{\begin{tikzpicture}[every node/.style={circle}]
	\node[draw] (0) {$0$} ;
	\node[draw=gray, right=1.2 cm of 0] (1) {\color{gray}$1$};
	\node[draw=gray, right=1.2 cm of 1] (2) {\color{gray}$2$};
	\node[draw=gray, right=1.2 cm of 2] (3) {\color{gray}$3$};

	\draw[->] (0) edge[loop below] node[auto] {$1$} (0);
	\draw[->,gray] (1) edge[] node[auto] {$0.5$} (0);
	\draw[->,gray] (1) edge[] node[below] {$0.5$} (2);
	\draw[->,gray] (2) edge[loop below] node[auto] {$1$} (2);
	\draw[->,gray] (3) edge[bend right=20] node[above, near start] {$0.5$} (0);
	\draw[->,gray] (3) edge[] node[below] {$0.5$} (2);
	\draw ($(0.west) + (-0.3,0)$) edge[->] (0);
	\draw  [use as bounding box, draw=white] (-1,-1.4) rectangle (6,1.3) {};
\end{tikzpicture}%
}}
	\subfigure[$D_{r_2}$ with $r_2(k_1)=1, r_2(k_2)=2$]{
	    \label{fig:realisation_1_2}
		\scalebox{0.7}{\begin{tikzpicture}[every node/.style={circle}]
	\node[draw] (0) {$0$};
	\node[draw, right=1.2 cm of 0] (1) {\color{black}$1$};
	\node[draw, right=1.2 cm of 1] (2) {\color{black}$2$};
	\node[draw, right=1.2 cm of 2] (3) {\color{gray}$3$};

	\draw[->] (0) edge[loop below] node[auto] {$0.5$} (0);
	\draw[->] (1) edge[loop below] node[auto] {$0.5$} (1);
	\draw[->] (2) edge[loop below] node[auto] {$1$} (2);
	\draw[->] (0) edge node[below] {$0.5$} (1);
	\draw[->] (1) edge node[below] {$0.5$} (2);
	\draw[->,gray] (3) edge[bend right=25] node[above, near start] {$0.5$} (1);
	\draw[->,gray] (3) edge[] node[below] {$0.5$} (2);
	\draw ($(0.west) + (-0.3,0)$) edge[->] (0);
	\draw  [use as bounding box, draw=white] (-1,-1.2) rectangle (6,1.3) {};
\end{tikzpicture}%
}}
	\subfigure[$D_{r_3}$ with $r_3(k_1)=1, r_3(k_2)=3$]{%
	    \label{fig:realisation_1_3}
		\scalebox{0.7}{\begin{tikzpicture}[every node/.style={circle}]
	\node[draw] (0) {$0$} ;
	\node[draw, right=1.2 cm of 0] (1) {\color{black}$1$};
	\node[draw=gray, right=1.2 cm of 1] (2) {\color{gray}$2$};
	\node[draw, right=1.2 cm of 2] (3) {$3$};
	
	\draw[->] (0) edge[loop below] node[auto] {$0.5$} (0);
	\draw[->] (1) edge[loop below] node[auto] {$0.5$} (1);
	\draw[->] (3) edge[loop below] node[auto] {$0.5$} (3);
	\draw[->] (0) edge node[auto] {$0.5$} (1);
	\draw[->] (1) edge[bend left] node[auto] {$0.5$} (3);
	\draw[->] (3) edge[bend left] node[auto] {$0.5$} (1);
	\draw[->,gray] (2) edge[] node[above, near start] {$1$} (3);

	\draw ($(0.west) + (-0.3,0)$) edge[->] (0);
	\draw  [use as bounding box, draw=white] (-1,-1.2) rectangle (6,1.3) {};
	\end{tikzpicture}}}
	\subfigure[$D_{r_4}$ with $r_4(k_1)=0, r_4(k_2)=3$]{%
	    \label{fig:realisation_0_3}
		\scalebox{0.7}{\begin{tikzpicture}[every node/.style={circle}]
	\node[draw] (0) {$0$} ;
	\node[draw, right=1.2 cm of 0] (1) {\color{gray}$1$};
	\node[draw=gray, right=1.2 cm of 1] (2) {\color{gray}$2$};
	\node[draw=gray, right=1.2 cm of 2] (3) {\color{gray}$3$};
	
	\draw[->] (0) edge[loop above] node[auto] {$1$} (0);
	\draw[->,gray] (1) edge[] node[auto] {$0.5$} (0);
	\draw[->,gray] (1) edge[bend right] node[below] {$0.5$} (3);
	\draw[->,gray] (2) edge[] node[above, near start] {$1$} (3);
	\draw[->,gray] (3) edge[loop above] node[auto] {$0.5$} (3);
	\draw[->,gray] (3) edge[bend right=30] node[above, near start] {$0.5$} (0);
		
	\draw ($(0.west) + (-0.3,0)$) edge[->] (0);
	\draw  [use as bounding box, draw=white] (-1,-1.2) rectangle (6,1.3) {};
	\end{tikzpicture}}}
		\vspace{-0.5em}
		\caption{The four different realisations of $\mathfrak{D}$.}
		
		\label{fig:realisation}
\end{figure}
We state two synthesis problems for families of MCs. 
The first is to identify the set of MCs satisfying and violating a given specification, respectively.
The second is to find a MC that maximises/minimises a given objective. 
We call these two problems \emph{threshold synthesis} and \emph{max/min synthesis}. 

\begin{myproblem}[Threshold synthesis]\label{prob:threshold} 
Let $\mathfrak{D}$ be a family of MCs and $\varphi$ a probabilistic reachability or expected reward specification.
The \emph{threshold synthesis problem} is to partition $\mathcal{R}^{\mathfrak{D}}$ into 
 $T$ and $F$ such that $\forall r \in T\colon D_r \vDash \varphi$ and $\forall r \in F\colon D_r \nvDash \varphi$.
\end{myproblem}
As a special case of the threshold synthesis problem, the \emph{feasibility synthesis problem} is to find just one realisation $r\in \mathcal{R}^{\mathfrak{D}}$ such that $D_r \vDash \varphi$.
\begin{myproblem}[Max synthesis]\label{prob:max} 
Let $\mathfrak{D}$ a family of MCs and $\phi=\lozenge G$ for $G\subseteq S$
The \emph{max synthesis problem} is to find a realisation  $r^* \in \mathcal{R}^{\mathfrak{D}}$ such that  $\texttt{\emph{Prob}}(D_{r^*}, \phi) = \max_{r\in \mathcal{R}_{\mathfrak{D}}} \{\texttt{\emph{Prob}}(D_{r}, \phi)\}$.
The problem is defined analogously for an expected reward measure or minimising realisations.
\end{myproblem}

\begin{example}[Synthesis problems]
	Recall the family of MCs $\mathfrak{D}$ from Example~\ref{ex:dtmc_family}.
	For the specification $\varphi=\mathbb{P}_{\geq 0.1}(\lozenge \{1\})$, the solution to the threshold synthesis problem is $T=\{r_2,r_3\}$ and $F=\{r_1,r_4\}$, as the goal state $1$ is not reachable for $D_{r_1}$ and $D_{r_4}$.
	For $\phi=\lozenge \{1\}$, the solution to the max synthesis problem on $\mathfrak{D}$ is $r_2$ or $r_3$, as $D_{r_2}$ and $D_{r_3}$ have probability one to reach state $1$.
\end{example}
\begin{mymethod}[One-by-one~\cite{DBLP:journals/fac/ChrszonDKB18}]\label{approach:one-by-one}
A straightforward solution to both synthesis problems is to enumerate all realisations $r\in\mathcal{R}^\mathfrak{D}$, model check the MCs $D_r$, and either compare all results with the given threshold or determine the maximum.
\end{mymethod}
We already saw that the number of realisations is exponential in $|K|$.
\begin{mytheorem}
		The feasibility synthesis problem is NP-complete.
\end{mytheorem}
The theorem even holds for almost-sure reachability properties. 
The proof is a straightforward adaption of results for augmented interval Markov chains~\cite[Theorem 3]{DBLP:journals/corr/Chonev17}, partial information games~\cite{DBLP:conf/hybrid/ChatterjeeKS13}, or partially observable MDPs~\cite{DBLP:conf/aaai/ChatterjeeCD16}.

\section{Guided Abstraction-Refinement Scheme}
In the previous section, we introduced the notion of a family of MCs, two synthesis problems and the one-by-one approach.
Yet, for a sufficiently high number of realisations such a straightforward analysis is not feasible.
We propose a novel approach allowing us to more efficiently analyse families of MCs.

\subsection{All-in-one MDP}\label{sec:all-in-one}

We first consider a single MDP that subsumes all individual MCs of a family $\mathfrak{D}$, and is equipped with an appropriate action and state labelling to identify the underlying realisations from $\mathcal{R}^{\mathfrak{D}}$.
\begin{mydef}[All-in-one MDP~\cite{GS13,RodriguesANLCSS15,DBLP:journals/fac/ChrszonDKB18}]\label{def:all-in-one}
The \emph{all-in-one MDP} of a family $\mathfrak{D} = (S,\init,K,\mathfrak{P})$ of MCs is given as $M^\mathfrak{D}=(S^\mathfrak{D},\init^\mathfrak{D},\Act^\mathfrak{D},\mathcal{P}^\mathfrak{D})$ where $S^\mathfrak{D}=S \times \mathcal{R}^{\mathfrak{D}}
\cup\{\init^\mathfrak{D}\}$, $\Act^\mathfrak{D}=\{\act^r\mid r\in\mathcal{R}^{\mathfrak{D}}\}$, and $\mathcal{P}^\mathfrak{D}$ is defined as follows: 
$$\mathcal{P}^\mathfrak{D}(\init^\mathfrak{D},a^r)((s_0,r)) = 1 \quad \text{and} \quad \mathcal{P}^\mathfrak{D}((s,r),a^r)((s',r))  = \mathfrak{P}(r)(s)(s').$$
\end{mydef}

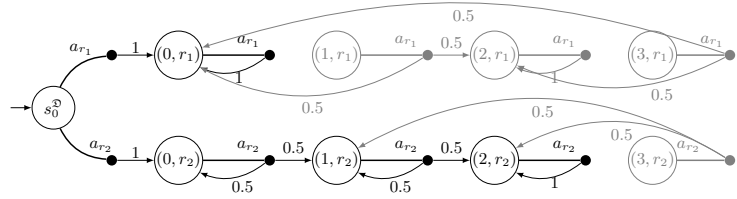
\begin{figure}[t]
	\centering
	\scalebox{0.7}{\begin{tikzpicture}[every node/.style={circle,inner sep=0.1pt}]
	\node[state] (s0) {$s_0^\mathfrak{D}$} ;
	\node[state,right=1.5 cm of s0,yshift=1cm] (0r1) {$(0,r_1)$};
	\node[state,right=1.5 cm of s0,yshift=-1cm] (0r2) {$(0,r_2)$};

	\node[state,right=4.5 cm of s0,yshift=1cm,draw=gray] (1r1) {\color{gray}$(1,r_1)$};
	\node[state,right=4.5 cm of s0,yshift=-1cm] (1r2) {$(1,r_2)$};
	
	\node[state,right=7.5 cm of s0,yshift=1cm,draw=gray] (2r1) {\color{gray}$(2,r_1)$};
	\node[state,right=7.5 cm of s0,yshift=-1cm] (2r2) {$(2,r_2)$};
	
	\node[state,right=10.5 cm of s0,yshift=1cm,draw=gray] (3r1) {\color{gray}$(3,r_1)$};
	\node[state,right=10.5 cm of s0,yshift=-1cm,draw=gray] (3r2) {\color{gray}$(3,r_2)$};

\iftrue
	\node[left=0.7 cm of 0r1,fill, circle, inner sep=2pt,] (s0ar1) {};
	\node[left=0.7 cm of 0r2,fill, circle, inner sep=2pt,] (s0ar2) {};
	
	\draw[-] (s0) edge[bend left, thick, near end] node[auto] {$\act_{r_1}$} (s0ar1);
	\draw[-] (s0) edge[bend right, thick, near end] node[auto] {$\act_{r_2}$} (s0ar2);

	\draw[->] (s0ar1) edge[] node[auto] {$1$} (0r1);
	\draw[->] (s0ar2) edge[] node[auto] {$1$} (0r2);
		
	\node[left=0.7 cm of 1r1,fill, circle, inner sep=2pt,] (0ar1) {};
	\node[left=0.7 cm of 1r2,fill, circle, inner sep=2pt,] (0ar2) {};
	
	\draw[-] (0r1) edge[thick, near end] node[auto] {$\act_{r_1}$} (0ar1);
	\draw[-] (0r2) edge[thick, near end] node[auto] {$\act_{r_2}$} (0ar2);

	\draw[->] (0ar1) edge[bend left] node[auto] {$1$} (0r1);
	\draw[->] (0ar2) edge[bend left] node[auto] {$0.5$} (0r2);
	\draw[->] (0ar2) edge[] node[auto] {$0.5$} (1r2);
	
	\node[left=0.7 cm of 2r1,fill=gray, circle, inner sep=2pt,] (1ar1) {};
	\node[left=0.7 cm of 2r2,fill, circle, inner sep=2pt,] (1ar2) {};
	
	\draw[-,gray] (1r1) edge[thick, near end] node[auto] {$\act_{r_1}$} (1ar1);
	\draw[-] (1r2) edge[thick, near end] node[auto] {$\act_{r_2}$} (1ar2);

	\draw[->,gray] (1ar1) edge[bend left=30] node[auto] {$0.5$} (0r1);
	\draw[->,gray] (1ar1) edge[] node[auto] {$0.5$} (2r1);
	\draw[->] (1ar2) edge[bend left] node[auto] {$0.5$} (1r2);
	\draw[->] (1ar2) edge[] node[auto] {$0.5$} (2r2);

	\node[left=0.7 cm of 3r1,fill=gray, circle, inner sep=2pt,] (2ar1) {};
	\node[left=0.7 cm of 3r2,fill, circle, inner sep=2pt,] (2ar2) {};
	
	\draw[-,gray] (2r1) edge[thick, near end] node[auto] {$\act_{r_1}$} (2ar1);
	\draw[-] (2r2) edge[thick, near end] node[auto] {$\act_{r_2}$} (2ar2);

	\draw[->,gray] (2ar1) edge[bend left] node[auto] {$1$} (2r1);
	\draw[->] (2ar2) edge[bend left] node[auto] {$1$} (2r2);

	\node[right=0.9 cm of 3r1,fill=gray, circle, inner sep=2pt,] (3ar1) {};
	\node[right=0.9 cm of 3r2,fill=gray, circle, inner sep=2pt,] (3ar2) {};
	
	\draw[-,gray] (3r1) edge[thick, near end] node[auto, near end] {$\act_{r_1}$} (3ar1);
	\draw[-,gray] (3r2) edge[thick, near start] node[auto] {$\act_{r_2}$} (3ar2);
%
	\draw[->,gray] (3ar1) edge[bend right=18] node[auto] {$0.5$} (0r1);
	\draw[->,gray] (3ar1) edge[bend left=25] node[below,pos=0.3] {$0.5$} (2r1);
	\draw[->,gray] (3ar2) edge[bend right=30] node[auto] {$0.5$} (1r2);
	\draw[->,gray] (3ar2) edge[bend right=30] node[auto] {$0.5$} (2r2);

	\draw ($(s0.west) - (0.4,0)$) edge[->] (s0);
\fi
	\end{tikzpicture}}\vspace{-0.3cm}	
	\caption{Reachable fragment of the all-in-one MDP $M^\mathfrak{D}$ for realisations $r_1$ and $r_2$.}
	\label{fig:all-in-one}
\end{figure}
\begin{example}[All-in-one MDP]\label{ex:all-in-one}
	Fig.~\ref{fig:all-in-one} shows the all-in-one MDP $M^\mathfrak{D}$ for the family $\mathfrak{D}$ of MCs from Example~\ref{ex:dtmc_family}. 
	Again, states that are not reachable from the initial state $\init^\mathfrak{D}$ are marked grey.
	For the sake of readability, we only include the transitions and states that correspond to realisations $r_1$ and $r_2$.
\end{example}
From the (fresh) initial state $\init^\mathfrak{D}$ of the MDP, the choice of an action $\act_r$ corresponds to choosing the realisation $r$ and entering the concrete MC $D_r$.
This property of the all-in-one MDP is formalised as follows.
\begin{mycorollary}
	For the all-in-one MDP $M^\mathfrak{D}$ of family $\mathfrak{D}$ of MCs:
	\[ \{ M^\mathfrak{D}_{\sigma^r} \mid \sigma^r \text{ memoryless deterministic scheduler } \} = \{ D_r \mid r \in \mathcal{R}^\mathfrak{D} \}\footnote{The original initial state $\init$ of the family of MCs needs to be the initial state of $M^\mathfrak{D}_{\sigma^r}$.}. \] 
\end{mycorollary}
Consequently, the feasibility synthesis problem for $\varphi$ has the solution $r\in\mathcal{R}^\mathfrak{D}$ iff there exists a memoryless deterministic scheduler $\sigma^r$ such that $M^\mathfrak{D}_{\sigma^r} \vDash \varphi$.
\begin{mymethod}[All-in-one~\cite{DBLP:journals/fac/ChrszonDKB18}]\label{approach:allinone}
Model checking the all-in-one MDP determines max or min probability (or expected reward) for all states, and thereby for all realisations, and thus provides a solution to both synthesis problems.
\end{mymethod}
As also the all-in-one MDP may be too large for realistic problems, we merely use it as formal starting point for our abstraction-refinement loop.

\subsection{Abstraction}\label{sec:abstraction}

First, we define a predicate abstraction that at each state of the MDP \emph{forgets} in which realisation we are, i.e., abstracts the second component of a state $(s, r)$. 
\begin{mydef}[Forgetting]\label{def:forgetting}
Let  $M^\mathfrak{D}=(S^\mathfrak{D},\init^\mathfrak{D},\Act^\mathfrak{D},\mathcal{P}^\mathfrak{D})$ be an all-in-one MDP. \emph{Forgetting} is an equivalence relation $\sim_f\ \subseteq S^\mathfrak{D} \times S^\mathfrak{D}$ satisfying  
 \[(s,r)\sim_f (s',r') \iff  s=s' \text{    and   }  s_0^{\mathfrak{D}} \sim_f (s_0^{\mathfrak{D}},r) \ \forall r\in  \mathcal{R}^\mathfrak{D}.\] 
Let $[s]_{\sim}$ denote the equivalence class wrt. $\sim_f$ containing~state~$s\in S^\mathfrak{D}$.
 
 Forgetting induces the \emph{quotient MDP} $M^\mathfrak{D}_\sim = ( \quotientstates,[\init^\mathfrak{D}]_\sim,\Act^\mathfrak{D},\mathcal{P}^\mathfrak{D}_\sim )$, where  
$\mathcal{P}^\mathfrak{D}_\sim([s]_\sim,a_r)([s']_\sim) = \mathfrak{P}(r)(s)(s')$.
\end{mydef}
At each state of the quotient MDP, the actions correspond to any realisation. 
It includes states that are unreachable in every realisation.

\begin{remark}[Action space]\label{rem:actions} 
According to Def.~\ref{def:forgetting}, for every state $[s]_\sim$  there are $|\mathfrak{D}|$ actions.
Many of these actions lead to the same distributions over successor states. 
In particular, two different realisations $r$ and $r'$ lead to the same distribution in $s$ if  $r(k) = r'(k)$ for all $k\in K$ where 
$\mathfrak{P}(s)(k) \neq 0$. 
To avoid this spurious blow-up of actions
, we \emph{a-priori} merge all actions yielding the same distribution.
\end{remark}
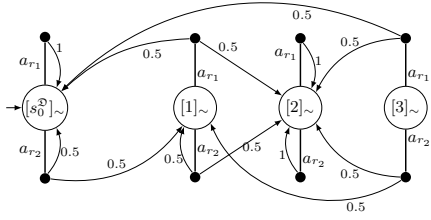
\begin{figure}[t]
	\centering
	\scalebox{0.7}{\begin{tikzpicture}[every node/.style={circle,inner sep=0.1pt}]
	\node[state] (s0) {$[s_0^\mathfrak{D}]_\sim$} ;
	
	\node[state,right=2 cm of s0] (1) {$[1]_\sim$};
	\node[state,right=4 cm of s0] (2) {$[2]_\sim$};
	\node[state,right=6 cm of s0] (3) {$[3]_\sim$};

	\node[above=0.8 cm of s0,fill, circle, inner sep=2pt,] (0ar1ar4) {};
	\node[below=0.8 cm of s0,fill, circle, inner sep=2pt,] (0ar2ar3) {};
	
	\draw[-] (s0) edge[thick] node[left] {$\act_{r_1}$} (0ar1ar4);
	\draw[-] (s0) edge[thick] node[left] {$\act_{r_2}$} (0ar2ar3);
	\draw[->] (0ar1ar4) edge[bend left] node[auto, near start] {\scriptsize{$1$}} (s0);
	\draw[->] (0ar2ar3) edge[bend right] node[right] {\scriptsize{$0.5$}} (s0);
	\draw[->] (0ar2ar3) edge[bend right] node[auto] {\scriptsize{$0.5$}} (1);
	
	\node[above=0.8 cm of 1,fill, circle, inner sep=2pt,] (1ar1) {};
	\node[below=0.8 cm of 1,fill, circle, inner sep=2pt,] (1ar2) {};
	
	\draw[-] (1) edge[thick, near start] node[right] {$\act_{r_1}$} (1ar1);
	\draw[-] (1) edge[thick, near start] node[auto, midway] {$\act_{r_2}$} (1ar2);%

	\draw[->] (1ar1) edge[bend right=20] node[below,near start] {\scriptsize{$0.5$}} (s0);
	\draw[->] (1ar1) edge[] node[auto, near start] {\scriptsize{$0.5$}} (2);
	\draw[->] (1ar2) edge[bend left] node[auto] {\scriptsize{$0.5$}} (1);
	\draw[->] (1ar2) edge[] node[right, midway] {\scriptsize{$0.5$}} (2);
%
	\node[above=0.8 cm of 2,fill, circle, inner sep=2pt,] (2ar1) {};
	\node[below=0.8 cm of 2,fill, circle, inner sep=2pt,] (2ar2) {};
	
	\draw[-] (2) edge[thick, near start] node[left, near end] {$\act_{r_1}$} (2ar1);
	\draw[-] (2) edge[thick, near start] node[right, near end] {$\act_{r_2}$} (2ar2);

	\draw[->] (2ar1) edge[bend left] node[auto] {\scriptsize{$1$}} (2);
	\draw[->] (2ar2) edge[bend left] node[auto] {\scriptsize{$1$}} (2);
%
	\node[above=0.8 cm of 3,fill, circle, inner sep=2pt,] (3ar1) {};
	\node[below=0.8 cm of 3,fill, circle, inner sep=2pt,] (3ar2) {};
	
	\draw[-] (3) edge[thick, near start] node[right] {$\act_{r_1}$} (3ar1);
	\draw[-] (3) edge[thick, near start] node[right] {$\act_{r_2}$} (3ar2);
	\draw[->] (3ar1) edge[bend right=35] node[near start,auto] {\scriptsize{$0.5$}} (s0);
	\draw[->] (3ar1) edge[bend right] node[above] {\scriptsize{$0.5$}} (2.north east);
	\draw[->] (3ar2) edge[bend left=50] node[near start,auto] {\scriptsize{$0.5$}} (1.south east);
	\draw[->] (3ar2) edge[bend left] node[below] {\scriptsize{$0.5$}} (2.south east);
%
	\draw ($(s0.west) + (-0.3,0)$) edge[->] (s0);
	\draw  [use as bounding box, draw=white] (-1,-2.1) rectangle (8.5,2.2) {};
	\end{tikzpicture}}\vspace{-0.5cm}
	\caption{The quotient MDP $M^\mathfrak{D}_\sim$ for realisations $r_1$ and $r_2$.}
	\label{fig:forgotten}
\end{figure}
The quotient MDP under forgetting involves that the available actions allow to switch realisations and thereby create induced MCs different from any MC in $\mathfrak{D}$. 
We formalise the notion of a consistent realisation with respect to parameters.
\begin{mydef}[Consistent realisation]
	For a family $\mathfrak{D}$ of MCs and $k\in K$, \emph{$k$-realisation-consistency} is an equivalence relation $\approx_k\ \subseteq \mathcal{R}^\mathfrak{D}{\times}\mathcal{R}^\mathfrak{D}$ satisfying:
\[
			r\approx_k r' \Longleftrightarrow r(k)=r'(k).
\]
	Let $[r]_{\approx_k}$ denote the equivalence class w.r.t. $\approx_k$ containing $r\in \mathcal{R}^\mathfrak{D}$.
\end{mydef}

\begin{mydef}[Consistent scheduler]\label{def:consistent-scheduler}
For quotient MDP $M^\mathfrak{D}_\sim$ after forgetting and $k\in K$, a scheduler $\sigma\in\Sigma^{M^\mathfrak{D}_\sim}$ is \emph{$k$-consistent} if for all 
$\pi,\pi'\in \pathsfin^{M^\mathfrak{D}_\sim}$: 
   		\[   			\sigma(\pi)=\act_r \land \sigma(\pi')=\act_{r'} \Longrightarrow r \approx_k r'\ .
   		\] 
A scheduler is  \emph{$K$-consistent}  (short: \emph{consistent}) if it is $k$-consistent for all $k\in K$.
\end{mydef}
\begin{lemma}
	For the quotient MDP $M^\mathfrak{D}_{\sim}$ of family $\mathfrak{D}$ of MCs:
	\[ \{ \left(M^\mathfrak{D}_{\sim}\right)_{\sigma^{r^*}} \mid \sigma^{r^*} \text{ consistent scheduler } \} = \{ D_r \mid r \in \mathcal{R}^\mathfrak{D} \}. \] 
\end{lemma}
\begin{proof}[Idea]
For $\sigma^r \in \Sigma^{M^\mathfrak{D}}$, we construct $\sigma^{r^*} \in \Sigma^{M^\mathfrak{D}_\sim}$  such that  $\sigma^{r^*}([s]_\sim) = \act_r$ for all $s$. Clearly  $\sigma^{r^*}$ is consistent and  $M^\mathfrak{D}_{\sigma^r} = \left(M^\mathfrak{D}_{\sim}\right)_{\sigma^{r^*}}$ is obtained via a map between $(s,r)$ and $[s]_\sim$.
 For  $\sigma^{r^*} \in \Sigma^{M^\mathfrak{D}_\sim}$, we construct $\sigma^r \in \Sigma^{M^\mathfrak{D}}$ such that if $\sigma^{r^*}([s]_\sim) = \act_r$ then  $\sigma^{r}(\init^{\mathfrak{D}}) = \act_r$. For all other states, we define $\sigma^{r}((s,r')) = a^{r'}$ independently of $\sigma^{r^*}$. 
 Then $M^\mathfrak{D}_{\sigma^r} = \left(M^\mathfrak{D}_{\sim}\right)_{\sigma^{r^*}}$  is obtained as above.
\end{proof}
The following theorem is a direct corollary: we need to consider exactly the consistent schedulers.
\begin{mytheorem}
For all-in-one MDP $M^\mathfrak{D}$ and specification $\varphi$, there exists a memoryless deterministic scheduler $\sigma^r \in 
\Sigma^{M^\mathfrak{D}}$ such that  $M^\mathfrak{D}_{\sigma^r} \vDash \varphi$ iff there exists a consistent deterministic scheduler $\sigma^{r^*}\in\Sigma^{M^\mathfrak{D}_\sim}$ such that  $\left(M^\mathfrak{D}_{\sim}\right)_{\sigma^{r^*}} \vDash \varphi$.
\end{mytheorem}
\begin{example}
	Recall the all-in-one MDP $M^\mathfrak{D}$ from Ex.~\ref{ex:all-in-one}.
	The quotient MDP $M^\mathfrak{D}_\sim$ is depicted in Fig.~\ref{fig:forgotten}. 
	Only the transitions according to realisations $r_1$ and $r_2$ are included.
	Transitions from previously unreachable states, marked grey in Ex.~\ref{ex:all-in-one}, are now available due to the abstraction.
	The scheduler $\sigma\in\Sigma^{M^\mathfrak{D}_\sim}$ with $\sigma([\init^\mathfrak{D}]_\sim)=\act_{r_2}$ and $\sigma([1]_\sim)=\act_{r_1}$ is \emph{not $k_1$-consistent} 
	as different values are chosen for $k_1$ by $r_1$ and $r_2$.
	In the MC $M^\mathfrak{D}_{\sim\sigma}$ induced by $\sigma$ and $M^\mathfrak{D}_\sim$, the probability to reach state $[2]_\sim$ is one, while under realisation $r_1$, state $2$ is not reachable.
	\end{example}
\begin{mymethod}[Scheduler iteration]\label{approach:scheduler-iteration}
Enumerating all consistent schedulers for $M^\mathfrak{D}_\sim$ and analysing the induced MC provides a solution to both synthesis problems.
\end{mymethod}
However, optimising over exponentially many consistent schedulers solves the NP-complete feasibility synthesis problem, rendering such an iterative approach unlikely to be efficient.
Another natural approach is to employ solving techniques for NP-complete problems, like satisfiability modulo linear real arithmetic.
\begin{mymethod}[SMT]\label{approach:SMT}
A dedicated SMT-encoding (in~Sect.~\ref{sec:smt}) of the induced MCs of consistent schedulers from $M^\mathfrak{D}_\sim$ that solves the feasibility problem. 
\end{mymethod}

\subsection{Refinement Loop}\label{sec:refinement}
Although iterating over consistent schedulers (Approach~\ref{approach:scheduler-iteration}) is not feasible, model checking of $M^\mathfrak{D}_\sim$ still provides useful information for the analysis of the family $\mathfrak{D}$. 
Recall the feasibility synthesis problem for $\varphi=\mathbb{P}_{\leq \lambda} (\phi)$. 
If $\texttt{Prob}^{\max}(M^\mathfrak{D}_\sim,\phi) \leq \lambda$,  then all realisations of $\mathfrak{D}$ satisfy $\varphi$. 
On the other hand, 
 $\texttt{Prob}^{\min}(M^\mathfrak{D}_\sim,\phi) > \lambda$ implies that there is no realisation satisfying $\varphi$. 
If $\lambda$ lies between the $\min$ and $\max$ probability, and the scheduler inducing the $\min$ probability is not consistent, we cannot conclude anything yet, i.e., the abstraction is too coarse. 
A natural countermeasure is to refine the abstraction represented by $M^\mathfrak{D}_\sim$,  in particular, split the set of realisations leading to two 
 synthesis sub-problems.

\begin{mydef}[Splitting]\label{def:splitRel} 
Let $\mathfrak{D}$ be a family of MCs, and $\mathcal{R} \subseteq \mathcal{R}^{\mathfrak{D}}$ a set of realisations. For $k\in K$ and predicate $A_k$ over~$S$, \emph{splitting} partitions  $\mathcal{R}$ into
\[\mathcal{R}_\top=\{ r \in \mathcal{R}  \mid A_k(r(k))\} \quad  \text{and} \quad \mathcal{R}_\bot=\{ r \in \mathcal{R}  \mid \neg A_k(r(k))\}.\]
\end{mydef}
Splitting the set of realisations, and considering the subfamilies separately, rather than splitting states in the quotient MDP,
is crucial for the performance of the synthesis process as we avoid rebuilding the quotient MDP in each iteration. Instead, we only restrict the actions of the MDP to the particular subfamily.
\begin{mydef}[Restricting]\label{def:res} 
Let $M^\mathfrak{D}_\sim= ( \quotientstates,[\init^\mathfrak{D}]_\sim,\Act^\mathfrak{D},\mathcal{P}^\mathfrak{D}_\sim )$ be a quotient MDP and $\mathcal{R} \subseteq \mathcal{R}^{\mathfrak{D}}$ a set of realisations. The \emph{restriction} of $M^\mathfrak{D}_\sim$ wrt. $\mathcal{R}$ is the MDP $M^\mathfrak{D}_\sim[\mathcal{R}] = ( \quotientstates,[\init^\mathfrak{D}]_\sim,\Act^\mathfrak{D}[\mathcal{R}],\mathcal{P}^\mathfrak{D}_\sim)$ where $\Act^\mathfrak{D}[\mathcal{R}] = \{a_r \mid r \in  \mathcal{R}\}.$\footnote{Naturally,  $\mathcal{P}^\mathfrak{D}_\sim$ in  $M^\mathfrak{D}_\sim[\mathcal{R}]$ is restricted to $\Act^\mathfrak{D}[\mathcal{R}]$.}
 \end{mydef}
The splitting operation is the core of the proposed abstraction-refinement.
Due to space constraints, we do not consider feasibility separately. 

Algorithm~\ref{alg:th} illustrates the \emph{threshold synthesis} process. 
Recall that the goal is to decompose the set $\mathcal{R}^{\mathfrak{D}}$ into realisations satisfying and violating a given specification, respectively.
\begin{algorithm}[t]
	\caption{Threshold synthesis}\label{alg:th}
	\hspace*{\algorithmicindent} \textbf{Input:} A family $\mathfrak{D}$ of MCs with the set $\mathcal{R}^\mathfrak{D}$ of realisations, and specification $\mathbb{P}_{\leq \lambda} (\phi)$  \\
    \hspace*{\algorithmicindent} \textbf{Output}: A partition of $\mathcal{R}^\mathfrak{D}$ into subsets $T$ and $F$ according to Problem~\ref{prob:threshold}. 
	\begin{algorithmic}[1]
	\State $F \gets \emptyset, \ T \gets \emptyset,  \ U \gets \{\mathcal{R}^\mathfrak{D}\}$
		\State  $M^\mathfrak{D}_\sim \gets \texttt{buildQuotientMDP}(\mathfrak{D},\mathcal{R}^\mathfrak{D},\sim_f)$  \Comment{Applying Def.~\ref{def:all-in-one}~and~\ref{def:forgetting}}
		\While{$U  \neq  \emptyset$}
		\State \textbf{select} $\mathcal{R} \in U$ \textbf{and} $\mathcal{U} \gets \mathcal{U} \setminus \{\mathcal{R}\}$
		\State $M^\mathfrak{D}_\sim[\mathcal{R}] \gets \texttt{restrict}(M^\mathfrak{D}_\sim,\mathcal{R})$  \Comment{Applying Def.~\ref{def:res}}
		\State $(\max,\sigma_{\max}) \gets \texttt{solveMaxMDP}(M^\mathfrak{D}_\sim[\mathcal{R}],\phi)$
		\State $(\min,\sigma_{\min}) \gets \texttt{solveMinMDP}(M^\mathfrak{D}_\sim[\mathcal{R}],\phi)$
		\If{$\max < \lambda$} $T \gets T \cup \mathcal{R}$
		\EndIf 
		\If{$\min > \lambda$} $F \gets F \cup \mathcal{R}$ \
		\EndIf 
		\If{$\min \leq  \lambda \leq \max$} 
		\State $U \gets U \cup \texttt{split}(\mathcal{R},\texttt{selPredicate}(\max,\sigma_{\max},\min,\sigma_{\min}))$    \Comment{See Sect.~\ref{sec:splitting}}
		\EndIf 
		\EndWhile\\
	\Return $T$, $F$
	\end{algorithmic}
	\end{algorithm}
The algorithm uses a set $U$ to store subfamilies of  $\mathcal{R}^{\mathfrak{D}}$ that have not been yet classified as  satisfying or violating. It starts building the quotient
MDP with merged actions. That is, we never construct the all-in-one MDP, and  we merge actions as discussed in Rem.~\ref{rem:actions}. 
For every $\mathcal{R} \in U$, the algorithm restricts the set of realisations to obtain the corresponding subfamily.
For the restricted quotient MDP, the algorithm runs standard MDP model checking to compute the $\max$ and $\min$ probability and corresponding schedulers, respectively. 
Then, the algorithm either classifies $\mathcal{R}$ as satisfying/violating, or splits it based on a suitable predicate, and updates $U$ accordingly. 
We describe the splitting strategy in the next subsection. 
The algorithm terminates if $U$ is empty, i.e., all subfamilies have been classified. 
As only a finite number of subfamilies of realisations has to be evaluated, termination is guaranteed.

The refinement loop for max synthesis is very similar, cf.~Alg.~\ref{alg:max}.
Recall that now the goal is to find the realisation $r^*$ that maximises the satisfaction probability $\max^*$ of a path formula. 
The difference between the algorithms lies in the interpretation of the results of the underlying MDP model checking. 
If the $\max$ probability for $\mathcal{R}$ is below $\max^*$, $\mathcal{R}$ can be discarded. 
Otherwise, we check whether the corresponding scheduler $\sigma_{\max}$  is consistent. 
If consistent, the algorithm updates $r^*$ and $\max^*$, and discards $\mathcal{R}$. 
If the scheduler is not consistent but $\min > \max^{*}$ holds, 
we can still update $\max^*$ and improve the pruning process, as it means that some realisation (we do not know which) in $\mathcal{R}$ induces a higher probability than 
$\max^*$. 
Regardless whether $\max^*$ has been updated, the algorithm has to split $\mathcal{R}$ based on some predicate, and analyse its subfamilies as they may include the maximising realisation.

\algdef{SE}[DOWHILE]{Do}{doWhile}{\algorithmicdo}[1]{\algorithmicwhile\ #1}%

\begin{algorithm}[t]
	\caption{Max synthesis}\label{alg:max}
	\hspace*{\algorithmicindent} \textbf{Input:} A family $\mathfrak{D}$  of MCs with the set $\mathcal{R}^\mathfrak{D}$ of realisations, and a path formula $\phi$  \\
    \hspace*{\algorithmicindent} \textbf{Output}: A realisation $r^* \in \mathcal{R}^\mathfrak{D}$  according to Problem~\ref{prob:max}. 
	\begin{algorithmic}[1]
		\State $\max^* \gets -\infty, \ U  \gets \{\mathcal{R}^\mathfrak{D}\}$
	\State  $M^\mathfrak{D}_\sim \gets \texttt{buildQuotientMDP}(\mathfrak{D},\mathcal{R}^\mathfrak{D},\sim_f)$  \Comment{Applying Def.~\ref{def:all-in-one}~and~\ref{def:forgetting}}
		\While{$U  \neq  \emptyset$}
\State \textbf{select} $\mathcal{R} \in U$ \textbf{and} $\mathcal{U} \gets \mathcal{U} \setminus \{\mathcal{R}\}$
		\State $M^\mathfrak{D}_\sim[\mathcal{R}] \gets \texttt{restrict}(M^\mathfrak{D}_\sim,\mathcal{R})$ \Comment{Applying Def.~\ref{def:res}}
		\State $(\max,\sigma_{\max}) \gets \texttt{solveMaxMDP}(M^\mathfrak{D}_\sim[\mathcal{R}],\phi)$
		\State $(\min,\sigma_{\min}) \gets \texttt{solveMinMDP}(M^\mathfrak{D}_\sim[\mathcal{R}],\phi)$
		\If{$\max > \max^*$} 
		\If{$\texttt{isConsistent}(\sigma_{\max})}$ $r^* \gets q_{\max}, \max^* \gets \max$ 
		\Else 
		\If{$\min > \max^*$}  $\max^* \gets \min$ 
		\EndIf
		\State $U \gets U \cup \texttt{split}(\mathcal{R},\texttt{selPredicate}(\max,\sigma_{\max},\min,\sigma_{\min}))$   \Comment{See Sect.~\ref{sec:splitting}}
		\EndIf
		\EndIf 
		\EndWhile \\
		\Return $r^{*}$
	\end{algorithmic}
\end{algorithm}

\subsection{Splitting strategies}\label{sec:splitting}
If verifying the quotient MDP $M^\mathfrak{D}_\sim[\mathcal{R}]$ cannot classify the (sub-)realisation $\mathcal{R}$ as satisfying or violating, we split $\mathcal{R}$, while we guide the splitting strategy by using the obtained verification results.
The splitting operation chooses a suitable parameter $k \in K$ and predicate $A_k$ that partition the realisations $\mathcal{R}$ into $\mathcal{R}_{\top}$ and $\mathcal{R}_{\bot}$ (see Def.~\ref{def:splitRel}).  
A good splitting strategy globally reduces the number of model-checking calls required to classify all $r\in \mathcal{R}$. 

The two key aspects to locally determine a good $k$ are: 1)~the \emph{variance}, that is, how the splitting may narrow the difference between $\max = \texttt{Prob}^{\max}(M^\mathfrak{D}_\sim[\mathcal{X}],\phi)$ and $\min = \texttt{Prob}^{\min}(M^\mathfrak{D}_\sim[\mathcal{X}],\phi)$ for both $\mathcal{X}  = \mathcal{R}_{\top}$ or $\mathcal{X}  = \mathcal{R}_{\bot}$, and 2)~the \emph{consistency}, that is, how the splitting may reduce the inconsistency of the schedulers  $\sigma_{\max}$~and~$\sigma_{\min}$. 
These aspects cannot be evaluated precisely without applying all the split operations and solving the new MDPs $M^\mathfrak{D}_\sim[\mathcal{R}_{\bot}]$ and $M^\mathfrak{D}_\sim[\mathcal{R}_{\top}]$. 
Therefore, we propose an efficient strategy that selects $k$ and $A_k$ based on a light-weighted analysis of the model-checking results for~$M^\mathfrak{D}_\sim[\mathcal{R}]$. 
The strategy applies two \emph{scores} $\texttt{variance}(k)$ and $\texttt{consistency}(k)$ that estimate  the influence of $k$ on the two key aspects.
For any $k$, the scores are accumulated over all \emph{important states} $s$ (reachable via $\sigma_{\max}$ or $\sigma_{\min}$, respectively) where $\mathfrak{P}(s)(k) \neq 0$.
A state $s$ is important for $\mathcal{R}$ and some $\delta \in \mathbb{R}_{\geq 0}$ if 
 \[ \frac{\texttt{Prob}^{\max}(M^\mathfrak{D}_\sim[\mathcal{R}],\phi)(s)  - \texttt{Prob}^{\min}(M^\mathfrak{D}_\sim[\mathcal{R}],\phi)(s)}{\texttt{Prob}^{\max}(M^\mathfrak{D}_\sim[\mathcal{R}],\phi)  - \texttt{Prob}^{\min}(M^\mathfrak{D}_\sim[\mathcal{R}],\phi)} \geq \delta \]
where $\texttt{Prob}^{\min}(.)(s)$ and $\texttt{Prob}^{\max}(.)(s)$ is the $\min$ and $\max$ probability in the MDP with  initial state $s$.
To reduce the overhead of computing the scores, we simplify the scheduler representation. In particular, for $\sigma_{\max}$ and every $k\in K$, we extract a map $C_{\max}^k\colon T_k \rightarrow \mathbb{N}$ , where $C_{\max}^k(t)$ is the number of important states for which $\sigma_{\max}(s) = a_r$ with $r(k) = t$.
The mapping~$C_{\min}^k$ represents~$\sigma_{\min}$.
 
We define $\texttt{variance}(k) = \sum_{t\in T_k} |C_{\max}^k(t) - C_{\min}^k(t)|$, leading to high scores if the two schedulers vary a lot. Further, we define  $ \texttt{consistency}(k) =  \texttt{size}\left(C_{\max}^k\right) \cdot \texttt{max}\left(C_{\max}^k\right) + \texttt{size}\left(C_{\min}^k\right) \cdot \texttt{max}\left(C_{\min}^k\right)$, where $\texttt{size}\left(C\right) = |\{t\in T_k \mid C(t) > 0\}|-1 $ and $\texttt{max}\left(C\right) = \max_{t\in T_k}\{ C(t)\}$, leading to high scores if the parameter has clear favourites   for $\sigma_{\max}$ and $\sigma_{\min}$, but values from its full range are chosen.

As indicated, we consider different strategies for the two synthesis problems. 
For threshold synthesis, we favour the impact on the variance as we principally do not need consistent schedulers. 
For the max synthesis, we favour the impact on the consistency, as we need a consistent scheduler inducing the $\max$ probability.

Predicate~$A_k$ is based on reducing the variance: 
The strategy selects  $T' \subset T_k$ with $|T'| = \frac{1}{2}\left \lceil{|T_k|}\right \rceil$, containing those~$t$ for which $C_{\max}^k(t) - C_{\min}^k(t)$ is the largest. The goal is to get a set of realisations that induce a large probability (the ones including $T'$ for parameter $k$) and the complement inducing a small~probability.

\begin{mymethod}[MDP-based abstraction refinement]\label{approach:absref}
		The methods underlying Algorithms~\ref{alg:th} and~\ref{alg:max}, together with the splitting strategies, provide solutions to the synthesis problems and are referred to as \emph{MDP abstraction} methods.
\end{mymethod}


\section{Experiments}


\label{sec:implementation}
We implemented the proposed synthesis methods as a Python prototype using Storm~\cite{DBLP:conf/cav/DehnertJK017}. In particular, we use the Storm Python API for model-adaption, -building, and -checking as well as for scheduler extraction. 
For SMT solving, we use Z3~\cite{DBLP:conf/tacas/MouraB08} via  pySMT~\cite{pysmt2015}.
The tool-chain takes a PRISM~\cite{KNP11} or JANI~\cite{DBLP:conf/tacas/BuddeDHHJT17} model with open integer constants, together with a set of expressions with possible values for these constants. 
The model may include the parallel composition of several modules/automata.
The open constants may occur in guards\footnote{slight care by the user is necessary to avoid deadlocks.}, probability definitions, and updates of the commands/edges.
Via adequate annotations, we identify the  parameter values that yield a particular action. The annotations are key to interpret the schedulers, and to restrict the quotient without rebuilding.
 
All experiments were executed on a Macbook~MF839LL/A with 8GB RAM memory limit and a 12h time out. All algorithms can significantly benefit from coarse-grained parallelisation, which we therefore do not consider here.

\subsection{Research questions and benchmarks} 
The  goal of the experimental evaluation is to answer the research question: \emph{How does the proposed MDP-based abstraction methods (Approaches 3--5) cope with the inherent complexity (i.e.~the NP-hardness) of the synthesis problems (cf.~Problems~1 and 2)?} To answer this question, we compare their performance  with Approaches 1 and 2~\cite{DBLP:journals/fac/ChrszonDKB18}, representing state-of-the-art solutions and the base-line algorithms.
The experiments show that the performance of the MDP abstraction significantly varies for different case studies. Thus, we consider benchmarks from various application domains to \emph{identify the key characteristics of the synthesis problems affecting the performance of our approach.}

\paragraph{Benchmarks description.}
We consider the following case studies: \textsl{Maze} is a planning problem typically considered as POMDP, e.g. in~\cite{DBLP:journals/rts/Norman0Z17}. The family describes all MCs induced by small-memory~\cite{DBLP:conf/aaai/ChatterjeeCD16,DBLP:conf/uai/Junges0WQWK018} observation-based deterministic strategies (with a fixed upper bound on the memory). We are interested in the expected time to the goal. In~\cite{DBLP:conf/uai/Junges0WQWK018}, parameter synthesis was used to find randomised strategies, using~\cite{DBLP:conf/atva/CubuktepeJJKT18}. \textsl{Pole} considers balancing a pole in a noisy and unknown environment (motivated by \cite{DBLP:conf/qest/ArmingBCKS18,DBLP:conf/aaai/ChadesCMNSB12}). At deploy time, the controller has a prior over a finite set of environment behaviours, and should optimise the expected behavior without depending on the actual (hidden) environment. The family describes schedulers that do not depend on the hidden information. We are interested in the expected time until failure.
\textsl{Herman} is an asynchronous encoding of the distributed Herman protocol for self-stabilising rings~\cite{DBLP:journals/ipl/Herman90,KNP12a}. The protocol is extended with a bit of memory for each station in the ring, and the choice to flip various unfair coins. Nodes in the ring are anonymous, they all behave equivalently (but may change their local memory based on local events).
The family describes variations of memory-updates and coin-selection, but preserves anonymity.
We are interested in the expected time until stabilisation.
\textsl{DPM} considers a partial information scheduler for a disk power manager motivated by~\cite{BBPM99,DBLP:journals/ase/GerasimouCT18}. We are interested in the expected energy consumption.
\textsl{BSN} (Body sensor network, \cite{RodriguesANLCSS15}) describes a network of connected sensors that identify health-critical situations. We are interested in the reliability.
The family contains various configurations of the used sensors. \textsl{BSN} is the largest software product line benchmark used in \cite{DBLP:journals/fac/ChrszonDKB18}. We drop some implications between features (parameters for us) as this is not yet supported by our modelling language. We thereby extended the family.



\begin{table}[t]
\caption{Benchmarks and timings for Approach 1--3}
\label{tab:benchmarklist}
\scriptsize
\begin{tabular}{llrr|rr|rrr|rr|r}
        &       &         &         & \multicolumn{2}{c|}{Member size} & \multicolumn{3}{c|}{Quotient size} & \multicolumn{3}{c}{Run time}                                               \\
Bench. & Range & $|K|$ & $|\mathcal{D}|$       & Avg. $|S|$   & Avg. $|T|$   & $|S|$    & $|A|$    & $|T|$   & 1-by-1 & All-in-1 &  \begin{tabular}[c]{@{}l@{}}Sched.\\ Enum.\end{tabular}   \\ \hline
   
\textsl{Pole}   & [3.35,3.82]   & 17      & 1327104 &  5689            &  16896            &   6793       & 7897         & 22416        &  130k$^{*}$     &    MO     &  26k          \\
\textsl{Maze} &   [9.8,9800] &  20  & 1048576 & 134 & 211 & 203 & 277 & 409 & 28k\phantom{$^{*}$} & TO &  2.7k \\
\textsl{Herman} &  [1.86,2.44] &   9 & 576 & 5287 & 6948 & 21313 & 102657 & 184096 & 55\phantom{$^{*}$} & 72 & 246\\ 
\textsl{DPM} & [68,210]  & 9   & 32768 & 5572 & 18147 & 35154 & 66096 & 160146 & 2.9k\phantom{$^{*}$}  & MO  & 7.2k \\
\textsl{BSN} & [0,0.988]  & 10   & 1024 & 116 & 196 & 382 & 457 & 762 & 31\phantom{$^{*}$} & 2  & 2 \\ 
\end{tabular}
\vspace{-4mm}
\end{table}

Table~\ref{tab:benchmarklist} shows the relevant statistics for each benchmark: the benchmark name, the (approximate) range of the $\min$ and $\max$ probability/reward for the given family, 
the number of non-singleton parameters $|K|$, and the number of family members $|\mathfrak{D}|$. 
Then, for the family members the average number of states and transitions of the MCs, and the states, actions ($= \sum_{s \in S} |\Act(s)|$), and transitions of the quotient MDP. 
Finally, it lists in seconds the run time of the base-line algorithms  and the consistent scheduler enumeration\footnote{Values with a $^{*}$ are estimated by sampling a large fraction of the family.}.
The base-line algorithms employ the one-by-one and the all-in-one technique, using either a BDD or a sparse matrix representation. We report the best results.
 MOs indicate breaking the memory limit. Only the all-in-one approach required significant memory.
As expected, the SMT-based implementation provides an inferior performance and thus we do not report its~results.

 
\subsection{Results and discussion}
To simplify the presentation, we focus primarily on the threshold synthesis problem as it allows a compact presentation of the key aspects.
Below, we provide some remarks about the performance for the max and feasibility synthesis.

\paragraph{Results.}

\begin{table}[t]
\caption{Results for threshold synthesis via abstraction-refinement}
\label{tab:thresholdsynt}
\scriptsize
\begin{tabular}{ll|rrrrr|rrrrrrr}
Inst &   $\lambda$ & \#\! \!Below & \begin{tabular}[c]{@{}l@{}}\#\! Subf \\ \ \ below\end{tabular} & \#\! \!Above & \begin{tabular}[c]{@{}l@{}}\#\! Subf \\ \ \ above\end{tabular} & Singles & \#\! \!Iter & Time & Build & Check & Anal. & Speedup \\\hline
  \multirow{5}{*}{\textsl{Pole}}     &  3.37 &  697 & 176 & 1326407 & 2186 & 920 & 4723 & 308 & 117 & 60 & 118 & \textbf{421}\\
     
    &    3.73       &  1307077        &   7854          &   20027       &  3279          &   1294            &    22265     &  1.7k     &  576     &  317     &  396 & \textbf{77}        \\
     &            3.76          &  1322181        &   3140          &   4923       &  1025          &      1022         &  8329    & 584     &       187 &   114    &   197    &   \textbf{222} \\
           &      3.79          &  1326502        &   572          &   602       &  123          &      74         &  1389    & 58     &       23 &   10    &   23    &   \textbf{2.2k}\\\hline
     \multirow{5}{*}{\textsl{Maze}} &  10 & 4 & 3 & 1048572 & 92 & 4 & 189 & 5 & $<$1 & 3 & $<1$ & \textbf{26k} \\
       &  20 & 4247 & 2297 & 1044329 & 4637 & 3400 & 13867 & 114 & 21  & 43 & 29 & \textbf{246} \\
       &  30 & 18188 & 9934 & 1030388 & 18004 & 14010 & 55875 & 608 & 80 & 127 & 270 & \textbf{46} \\
              &  8000 & 1046285 & 846 & 2291 & 1125 & 969 & 3941 & 136 & 9 & 106 & 13 & \textbf{1.0k} \\
      
      \hline
   \multirow{2}{*}{\textsl{Herman}}  & 1.9 & 6 & 6 & 570 & 368 & 320 & 747 & 333 & 303 & 11 & 18 & \textbf{0.2} \\
    &   1.71 & 0 & 0 & 576 & 258 & 184 & 515 & 232 & 206 & 8 & 17 & \textbf{0.3} \\\hline
      \multirow{3}{*}{\textsl{DPM}} &   80 &160 & 141 & 32608 & 1292 &  356 & 2865 & 1.0k & 602 & 322 & 64 & \textbf{3} \\
      &  70 &6 & 6 & 32762 & 443 & 40  & 897 & 380 & 190 & 156 & 32 & \textbf{8} \\
            & 60 &0  & 0 & 32768 &  104 &  6 & 207 & 99 & 42 & 48 & 8 & \textbf{29} \\\hline
             \multirow{2}{*}{\textsl{BSN}} &   .965 & 544 & 81 & 480 & 81 &  25 & 321 & 2 & $<$1 & $<$1 & $<$1 & \textbf{1} \\
               &    .985 & 994 & 41 & 30 & 8 &  5 & 97 & $<$1 & $<$1 & $<$1 & $<$1 & \textbf{3} \\\hline
\end{tabular}
\end{table}

Table~\ref{tab:thresholdsynt} shows results for threshold synthesis.
The first two columns indicate the benchmark and the various thresholds. For each threshold $\lambda$, the table lists the number of family members below (above) $\lambda$, each with the number of subfamilies that together contain these instances, and the number of singleton subfamilies that were considered.
 The last table part gives the number of iterations of the loop in Alg.~\ref{alg:th}, and timing information (total, build/restrict times, model checking times, scheduler analysis times). 
The last column gives the speed-up over the best base-line (based on the estimates).

\paragraph{Key observations.}
The speed-ups drastically vary, which shows that the MDP abstraction often achieves a superior performance but may also lead to a performance degradation in some cases.
We identify four key factors.


\smallskip
\noindent\textbf{Iterations.} As typical for CEGAR approaches, the key characteristic of the benchmark that affects the performance is the number $N$ of iterations in the refinement loop. 
The abstract action introduces an overhead per iteration caused by performing two MDP verification calls and by the scheduler analysis. 
The run time for \textsl{BSN}, with a small $|\mathfrak{D}|$ is actually significantly affected by the initialisation of various data structures; thus only a small speedup is achieved.

\smallskip
\noindent\textbf{Abstraction size.} 
The size of the quotient, compared to the average size of the family members, is relevant too.
The quotient includes at least all reachable states of all family members, and may be significantly larger if an inconsistent scheduler reaches states which are unreachable under any consistent scheduler.
The existence of such states is a common artefact from encoding families in high-level languages.
Table~\ref{tab:benchmarklist}, however, indicates that we obtain a very compact representation  for \textsl{Maze} and  \textsl{Pole}.

\smallskip
\noindent\textbf{Thresholds.}
The most important aspect is the threshold $\lambda$. 
If  $\lambda$ is closer to the optima, the abstraction requires a smaller number of iterations, which directly improves the performance.
We emphasise that in various domains, thresholds that ask for close-to-optimal solutions are indeed of highest relevance as they typically represent the system designs developers are most interested in~\cite{skaf2010techniques}.
\emph{Why do thresholds affect the number of iterations?}
Consider a family with $T_k = \{0,1\}$ for each~$k$. Geometrically, the set $\mathcal{R}^\mathfrak{D}$ can be visualised as $|K|$-dimensional cube. 
The cube-vertices reflect family members. Assume for simplicity that one of these vertices is optimal with respect to the specification.
Especially in benchmarks where parameters are equally important, the induced probability of a vertex roughly corresponds to the Manhattan distance to the optimal vertex. 
Thus, vertices above the threshold induce a diagonal hyperplane, which our splitting method approximates with orthogonal splits. Splitting diagonally is not possible, as it would induce optimising over observation-based schedulers.
Consequently, we need more and more splits the more the diagonal goes through the middle of the cube. 
\emph{Even when splitting optimally, there is a combinatorial blow-up in the required splits when the threshold is further from the optimal values.}
Another effect is that thresholds far from optima are more affected by the over-approximation of the MDP model-checking results and thus yield more inconclusive answers.

\smallskip
\noindent\textbf{Refinement strategy.}
So far, we reasoned about optimal splits.
Due to the computational overhead,  our strategy cannot ensure optimal splits. Instead, the strategy depends mostly on information encoded in the computed MDP strategies. 
\emph{In models where the optimal parameter value heavily depends on the state, the obtained schedulers are highly inconsistent and carry only limited information for splitting.} Consequently, in such benchmarks we split sub-optimally. The sub-optimality has a major impact on the performance for \emph{Herman} as all obtained strategies are highly inconsistent -- they take a different coin for each node, which is good to speed up the stabilisation of the ring.

\paragraph{Summary.}
MDP abstraction is not a silver bullet. It has a lot of potential in threshold synthesis when the threshold is close to the optima.
Consequently, \emph{feasibility synthesis with unsatisfiable specifications is handled perfectly well by MDP abstraction}, while this is the worst-case for enumeration-based approaches.
Likewise, \emph{max synthesis} can be understood as threshold synthesis with a shifting threshold $\max^{*}$:
If the $\max^{*}$ is quickly set close to $\max$, MDP abstraction yields superior performance. Roughly, we can quickly approximate $\max^{*}$ when some of the parameter values are clearly beneficial for the specification.

\section{Conclusion and Future Work}

We contributed to the efficient analysis of families of Markov chains.
In particular, we discussed and implemented existing approaches to solve practically interesting synthesis problems, and devised a novel abstraction refinement scheme that 
mitigates the computational complexity of the synthesis problems, as shown by the empirical evaluation.
In the future, we will include refinement strategies based on counterexamples as in~\cite{DBLP:journals/scp/JansenWAZKBS14,DBLP:conf/atva/DehnertJWAK14}.
%
%
\bibliographystyle{splncs04}
\bibliography{literature}

\begin{thebibliography}{10}
\providecommand{\url}[1]{\texttt{#1}}
\providecommand{\urlprefix}{URL }
\providecommand{\doi}[1]{https://doi.org/#1}

\bibitem{additionalmaterial}
Additional material, \url{https://github.com/moves-rwth/shepherd}

\bibitem{DBLP:conf/qest/ArmingBCKS18}
Arming, S., Bartocci, E., Chatterjee, K., Katoen, J.P., Sokolova, A.:
  Parameter-independent strategies for {pMDPs} via {POMDPs}. In: {QEST}. LNCS,
  vol. 11024, pp. 53--70. Springer (2018)

\bibitem{aziz1995usually}
Aziz, A., Singhal, V., Balarin, F., Brayton, R.K., Sangiovanni-Vincentelli,
  A.L.: It usually works: The temporal logic of stochastic systems. In: CAV.
  LNCS, vol.~939, pp. 155--165. Springer (1995)

\bibitem{DBLP:reference/mc/BaierAFK18}
Baier, C., de~Alfaro, L., Forejt, V., Kwiatkowska, M.: Model checking
  probabilistic systems. In: Handbook of Model Checking, pp. 963--999. Springer
  (2018)

\bibitem{BK08}
Baier, C., Katoen, J.: Principles of Model Checking. MIT Press (2008)

\bibitem{bartocci2011model}
Bartocci, E., Grosu, R., Katsaros, P., Ramakrishnan, C., Smolka, S.: Model
  repair for probabilistic systems. In: TACAS, LNCS, vol.~6605, pp. 326--340.
  Springer (2011)

\bibitem{BBPM99}
Benini, L., Bogliolo, A., Paleologo, G., Micheli, G.D.: Policy optimization for
  dynamic power management. IEEE Transactions on Computer-Aided Design of
  Integrated Circuits and Systems  \textbf{8}(3),  299--316 (2000)

\bibitem{DBLP:conf/tacas/BuddeDHHJT17}
Budde, C.E., Dehnert, C., Hahn, E.M., Hartmanns, A., Junges, S., Turrini, A.:
  {JANI:} quantitative model and tool interaction. In: TACAS. LNCS, vol. 10206,
  pp. 151--168 (2017)

\bibitem{CALINESCU2018140}
Calinescu, R., \v{C}e\v{s}ka, M., Gerasimou, S., Kwiatkowska, M., Paoletti, N.:
  Efficient synthesis of robust models for stochastic systems. Journal of
  Systems and Software  \textbf{143},  140 -- 158 (2018)

\bibitem{Ceska2017}
{\v{C}}e{\v{s}}ka, M., Dannenberg, F., Paoletti, N., Kwiatkowska, M., Brim, L.:
  Precise parameter synthesis for stochastic biochemical systems. Acta
  Informatica  \textbf{54}(6),  589--623 (2017)

\bibitem{DBLP:conf/aaai/ChadesCMNSB12}
Chades, I., Carwardine, J., Martin, T.G., Nicol, S., Sabbadin, R., Buffet, O.:
  {MOMDPs}: {A} solution for modelling adaptive management problems. In:
  {AAAI}. {AAAI} Press (2012)

\bibitem{DBLP:conf/cav/ChasinsP17}
Chasins, S., Phothilimthana, P.M.: Data-driven synthesis of full probabilistic
  programs. In: {CAV}. LNCS, vol. 10426, pp. 279--304. Springer (2017)

\bibitem{DBLP:conf/aaai/ChatterjeeCD16}
Chatterjee, K., Chmelik, M., Davies, J.: A symbolic {SAT}-based algorithm for
  almost-sure reachability with small strategies in {POMDPs}. In: {AAAI}. pp.
  3225--3232. {AAAI} Press (2016)

\bibitem{DBLP:conf/hybrid/ChatterjeeKS13}
Chatterjee, K., K{\"{o}}{\ss}ler, A., Schmid, U.: Automated analysis of
  real-time scheduling using graph games. In: {HSCC}. pp. 163--172. {ACM}
  (2013)

\bibitem{DBLP:conf/tase/ChenHHKQ013}
Chen, T., Hahn, E.M., Han, T., Kwiatkowska, M.Z., Qu, H., Zhang, L.: Model
  repair for {M}arkov decision processes. In: {TASE}. pp. 85--92. {IEEE} (2013)

\bibitem{DBLP:journals/corr/Chonev17}
Chonev, V.: Reachability in augmented interval {M}arkov chains. CoRR
  \textbf{abs/1701.02996} (2017)

\bibitem{DBLP:journals/fac/ChrszonDKB18}
Chrszon, P., Dubslaff, C., Kl{\"{u}}ppelholz, S., Baier, C.: Pro{F}eat:
  feature-oriented engineering for family-based probabilistic model checking.
  Formal Asp. Comput.  \textbf{30}(1),  45--75 (2018)

\bibitem{DBLP:journals/sttt/ClassenCHLS12}
Classen, A., Cordy, M., Heymans, P., Legay, A., Schobbens, P.: Model checking
  software product lines with {SNIP}. {STTT}  \textbf{14}(5),  589--612 (2012)

\bibitem{DBLP:journals/scp/ClassenCHLS14}
Classen, A., Cordy, M., Heymans, P., Legay, A., Schobbens, P.: Formal
  semantics, modular specification, and symbolic verification of product-line
  behaviour. Sci. Comput. Program.  \textbf{80},  416--439 (2014)

\bibitem{DBLP:conf/sigsoft/CordyHLSDL14}
Cordy, M., Heymans, P., Legay, A., Schobbens, P.Y., Dawagne, B., Leucker, M.:
  Counterexample guided abstraction refinement of product-line behavioural
  models. In: {SIGSOFT} {FSE}. pp. 190--201. {ACM} (2014)

\bibitem{DBLP:conf/atva/CubuktepeJJKT18}
Cubuktepe, M., Jansen, N., Junges, S., Katoen, J.P., Topcu, U.: Synthesis in
  p{MDP}s: {A} tale of 1001 parameters. In: {ATVA}. LNCS, vol. 11138, pp.
  160--176. Springer (2018)

\bibitem{DBLP:conf/atva/DehnertJWAK14}
Dehnert, C., Jansen, N., Wimmer, R., {\'{A}}brah{\'{a}}m, E., Katoen, J.P.:
  Fast debugging of {PRISM} models. In: {ATVA}. LNCS, vol.~8837, pp. 146--162.
  Springer (2014)

\bibitem{dehnert2015prophesy}
Dehnert, C., Junges, S., Jansen, N., Corzilius, F., Volk, M., Bruintjes, H.,
  Katoen, J.P., {\'A}brah{\'a}m, E.: {PROPhESY: A PRObabilistic ParamEter
  SYNnthesis Tool}. In: CAV. LNCS, vol.~9206, pp. 214--231. Springer (2015)

\bibitem{DBLP:conf/cav/DehnertJK017}
Dehnert, C., Junges, S., Katoen, J.P., Volk, M.: A storm is coming: {A} modern
  probabilistic model checker. In: CAV. LNCS, vol. 10427, pp. 592--600.
  Springer (2017)

\bibitem{pysmt2015}
Gario, M., Micheli, A.: Pysmt: a solver-agnostic library for fast prototyping
  of {SMT}-based algorithms. In: SMT Workshop 2015 (2015)

\bibitem{DBLP:journals/ase/GerasimouCT18}
Gerasimou, S., Calinescu, R., Tamburrelli, G.: Synthesis of probabilistic
  models for quality-of-service software engineering. Autom. Softw. Eng.
  \textbf{25}(4),  785--831 (2018)

\bibitem{GS13}
Ghezzi, C., Sharifloo, A.M.: Model-based verification of quantitative
  non-functional properties for software product lines. Information {\&}
  Software Technology  \textbf{55}(3),  508--524 (2013)

\bibitem{DBLP:journals/tcs/GiroDF14}
Giro, S., D'Argenio, P.R., Fioriti, L.M.F.: Distributed probabilistic
  input/output automata: Expressiveness, (un)decidability and algorithms.
  Theor. Comput. Sci.  \textbf{538},  84--102 (2014)

\bibitem{DBLP:conf/atva/GiroR12}
Giro, S., Rabe, M.N.: Verification of partial-information probabilistic systems
  using counterexample-guided refinements. In: {ATVA}. LNCS, vol.~7561, pp.
  333--348. Springer (2012)

\bibitem{hahn2011probabilistic}
Hahn, E.M., Hermanns, H., Zhang, L.: Probabilistic reachability for parametric
  {M}arkov models. Software Tools for Technology Transfer  \textbf{13}(1),
  3--19 (2011)

\bibitem{hansson_pctl}
Hansson, H., Jonsson, B.: A logic for reasoning about time and reliability.
  Formal Aspects of Computing  \textbf{6}(5),  512--535 (1994)

\bibitem{DBLP:journals/ipl/Herman90}
Herman, T.: Probabilistic self-stabilization. Inf. Process. Lett.
  \textbf{35}(2),  63--67 (1990)

\bibitem{DBLP:journals/scp/JansenWAZKBS14}
Jansen, N., Wimmer, R., {\'{A}}brah{\'{a}}m, E., Zajzon, B., Katoen, J.P.,
  Becker, B., Schuster, J.: Symbolic counterexample generation for large
  discrete-time {M}arkov chains. Sci. Comput. Program.  \textbf{91},  90--114
  (2014)

\bibitem{DBLP:conf/uai/Junges0WQWK018}
Junges, S., Jansen, N., Wimmer, R., Quatmann, T., Winterer, L., Katoen, J.P.,
  Becker, B.: Finite-state controllers of {POMDPs} using parameter synthesis.
  In: {UAI}. pp. 519--529. {AUAI} Press (2018)

\bibitem{Koc2015}
Kochenderfer, M.J.: Decision Making Under Uncertainty: Theory and Application.
  The MIT Press, 1st edn. (2015)

\bibitem{KNP12a}
Kwiatkowska, M., Norman, G., Parker, D.: Probabilistic verification of
  herman’s self-stabilisation algorithm. Formal Aspects of Computing
  \textbf{24}(4),  661--670 (2012)

\bibitem{KNP11}
Kwiatkowska, M., Norman, G., Parker, D.: \textsc{Prism} 4.0: Verification of
  probabilistic real-time systems. In: CAV. LNCS, vol.~6806, pp. 585--591.
  Springer (2011)

\bibitem{DBLP:conf/tacas/MouraB08}
de~Moura, L.M., Bj{\o}rner, N.: {Z3:} an efficient {SMT} solver. In: {TACAS}.
  LNCS, vol.~4963, pp. 337--340. Springer (2008)

\bibitem{DBLP:conf/pldi/NoriORV15}
Nori, A.V., Ozair, S., Rajamani, S.K., Vijaykeerthy, D.: Efficient synthesis of
  probabilistic programs. In: {PLDI}. pp. 208--217. {ACM} (2015)

\bibitem{DBLP:journals/rts/Norman0Z17}
Norman, G., Parker, D., Zou, X.: Verification and control of partially
  observable probabilistic systems. Real-Time Systems  \textbf{53}(3),
  354--402 (2017)

\bibitem{pathak-et-al-nfm-2015}
Pathak, S., {\'{A}}brah{\'{a}}m, E., Jansen, N., Tacchella, A., Katoen, J.P.: A
  greedy approach for the efficient repair of stochastic models. In: {NFM}.
  LNCS, vol.~9058, pp. 295--309. Springer (2015)

\bibitem{RodriguesANLCSS15}
Rodrigues, G.N., Alves, V., Nunes, V., Lanna, A., Cordy, M., Schobbens, P.,
  Sharifloo, A.M., Legay, A.: Modeling and verification for probabilistic
  properties in software product lines. In: {HASE}. pp. 173--180. {IEEE} (2015)

\bibitem{skaf2010techniques}
Skaf, J., Boyd, S.: Techniques for exploring the suboptimal set. Optimization
  and Engineering  \textbf{11}(2),  319--337 (2010)

\bibitem{DBLP:conf/fm/VandinBLL18}
Vandin, A., ter Beek, M.H., Legay, A., Lluch{-}Lafuente, A.: Qflan: {A} tool
  for the quantitative analysis of highly reconfigurable systems. In: {FM}.
  LNCS, vol. 10951, pp. 329--337. Springer (2018)

\bibitem{DBLP:conf/splc/VarshosazK13}
Varshosaz, M., Khosravi, R.: Discrete time {M}arkov chain families: modeling
  and verification of probabilistic software product lines. In: {SPLC}
  Workshops. pp. 34--41. {ACM} (2013)

\end{thebibliography}

\clearpage
\appendix
\section*{Appendix}
\section{SMT encoding of the  feasibility synthesis problem}\label{sec:smt}

Consider the quotient MDP $M^\mathfrak{D}_\sim = ( \quotientstates,[\init^\mathfrak{D}]_\sim,\Act^\mathfrak{D},\mathcal{P}^\mathfrak{D}_\sim )$, and specification $\varphi$.
W.l.o.g., $\varphi=\mathbb{E}_{\leq \kappa} (\lozenge G)$ where $G\subseteq \quotientstates$ and $\kappa\in \mathbb{R}_{\geq 0}$.
We define the typical set of states where there exists a scheduler inducing probability zero to reach $G$ as $\texttt{pZeroE}(\quotientstates)=\{s\in \quotientstates \mid \exists\sigma\in\Sigma^{M^\mathfrak{D}_\sim}\colon\texttt{Prob}(M^\mathfrak{D}_{\sim\sigma},\lozenge G)(s)=0)\}$, and the set of states where the probability to reach $G$ is one for all schedulers is 
$\texttt{pOneA}(\quotientstates)=\{s\in \quotientstates \mid \forall\sigma\in\Sigma^{M^\mathfrak{D}_\sim}\colon\texttt{Prob}(M^\mathfrak{D}_{\sim\sigma},\lozenge G)(s)=1)\}$.
Both sets are determined by a simple graph-analysis~\cite{BK08}.
Recall that the probability to reach $G$ under a certain scheduler needs to be one to have the expected reward well-defined.
We need to ensure that property for \emph{relevant states} $S_{\mathit{rel}}= \quotientstates\setminus\, \texttt{pOneA}(\quotientstates)$.
We refer to $S_{\mathit{crit}}=\texttt{pZeroE}(\quotientstates)$ as \emph{critical states} because we need to ensure that we choose a scheduler that does not induce zero probability.
The \emph{successor states} of a state $s$ w.r.t to an action $\act$ are defined as $\successors(s,\act)=\supp(\mathcal{P}^\mathfrak{D}_\sim(s,a))$.

For the encoding, we utilise the following variables:
\begin{compactitem}
	\item $e_s\in \mathbb{R}_{\geq 0}$ for all $s\in \quotientstates$; the expected reward to reach $G$ from state $s$.
	\item $\sigma_s^\act\in\mathbb{B}$ for all $s\in \quotientstates$ and $\act\in \Act^\mathfrak{D}$; $\sigma_s^\act=$\true iff action $\act$ is chosen the current scheduler at state $s$.
	\item $\probOneG_s\in\mathbb{B}$ for all $s\in \quotientstates$; $\probOneG_s=$\true iff the probability to reach $G$ is one under the current scheduler.
	\item $\probPosG_s\in\mathbb{B}$ for all $s\in \quotientstates$, $\probPosG_s = $\true iff the probability to reach $G$ is positive under the current scheduler.
	\item $o_s\in\mathbb{R}$ for all $s\in \quotientstates$; these variables are used to create a partial order on the states to ensure that the scheduler does not create a deadlock.
\end{compactitem}
\begin{align}
				    	&\quad e_{[\init^\mathfrak{D}]_\sim}\leq \kappa \land \probOneG_{s_0} \label{eq:threshold}\\
	\forall s\in G. &\quad r_s=0\label{eq:goal}\\
	\forall s\in \quotientstates.\,\forall\act\in\Act^\mathfrak{D}(s). &\quad \sigma_s^\act \rightarrow  \bigl(e_s \geq \mathit{rew}(s) + \sum_{s'\in \successors(s,a)} \mathcal{P}^\mathfrak{D}_{\sim}(s,\act)(s') \cdot e_{s'}\bigr)\label{eq:rewcomputation}\\
	\forall s\in \quotientstates. &\quad \bigoplus_{\act\in \Act^\mathfrak{D}(s)}\sigma_s^\act \land \Big(\bigwedge_{s'\in S^\mathfrak{D}} \sigma_s^{\act_r}\land\sigma_{s'}^{\act_{r'}} \rightarrow r \approx_{\mathcal{R}}^k r' \Big)\label{eq:k-consistent}\\
	 \forall s\in S_{\mathit{rel}}.\forall\act\in\Act^\mathfrak{D}(s). &\quad \sigma^\act_s \rightarrow \bigl(\probOneG_s \leftrightarrow \bigwedge_{s'\in\successors(s,\act)} \probOneG_{s'} \land  \probPosG_{s}\bigr) \label{eq:prob1}\\
	 \forall s\in S_{\mathit{crit}}.\forall\act\in\Act^\mathfrak{D}(s). &\quad \sigma^\act_s \rightarrow \bigl(\probPosG_s \leftrightarrow \bigvee_{s'\in\successors(s,\act)} \probPosG_{s'} \land  o_s<o_{s'}\bigr) \label{eq:probpos}
\end{align}

For parameter $k\in K$, the SMT encoding yields a $k$-consistent scheduler that induces a solution to the feasibility synthesis problem, if one exists.
\eqref{eq:threshold} ensures the expected reward to reach $G$ is at most $\kappa$, \eqref{eq:goal} sets the expected reward at goal states to zero.
\eqref{eq:rewcomputation} computes an upper bound on the expected reward for each state and action pair if 
the scheduler chooses $\act$ at $s$.
\eqref{eq:k-consistent} sets exactly one $\sigma_s^\act$ to true for each state, and ensures that the scheduler is $k$-consistent.
\eqref{eq:prob1} $\probOneG_{s}$ \true, indicating that $G$ is almost surely reached, 
 iff all its successors for action $\act$ are \true, and the probability to reach $G$ from $s$ is positive.
Analogously, \eqref{eq:probpos} is \true, indicating that $G$ is reached from $s$ with positive probability, iff there is a successor which also has positive probability to reach $G$.
The constraints on the $o_s$ variables ensure that a scheduler cannot generate a deadlock, as each successor $s'$ of state $s$ needs to have a higher value for $o_{s'}$.

\end{document}